\documentclass[review]{elsarticle}

\usepackage{amsmath,amsthm,amssymb,latexsym,stmaryrd,graphicx}
\usepackage{tikz}
\usepackage{tikz-cd}
\usepackage[mathscr]{euscript} 
\usepackage{stackrel}
\usepackage[all]{xy}

\journal{Journal of \LaTeX\ Templates}

\begin{document}

\begin{frontmatter}

\title{Towards a Homotopy Domain Theory\tnoteref{mytitlenote}}


\author{Daniel O. Martínez-Rivillas}

\author{Ruy J.G.B. de Queiroz}


\begin{abstract}
An appropriate framework is put forward for the construction of $\lambda$-models with $\infty$-groupoid structure, which we call \textit{homotopic $\lambda$-models}, through the use of an $\infty$-category with cartesian closure and enough points. With this, we establish the start of a project of generalization of Domain Theory and $\lambda$-calculus, in the sense that the concept of proof (path) of equality of $\lambda$-terms is raised to \textit{higher proof} (homotopy).
\end{abstract}


\end{frontmatter}


\newtheorem{defin}{Definition}[section]
\newtheorem{teor}{Theorem}[section]
\newtheorem{corol}{Corollary}[section]
\newtheorem{prop}{Proposition}[section]
\newtheorem{rem}{Remark}[section]
\newtheorem{lem}{Lemma}[section]
\newtheorem{nota}{Notation}[section]
\newtheorem{ejem}{Example}[section]

\section{Introduction}

The purpose of this paper is to give a framework for building a lambda model endowed with a topology, such that any proof of $\beta$-equality between $\lambda$-terms is not represented by equality between points (extensional equality), rather by the existence of a continuous path between the terms (intensional equality), where the interpretation of these terms corresponds to two points in the space. As an example, given the $\beta$-equality between different $\lambda$-terms
$$(\lambda x.(\lambda y.yx)z)v=_{\beta}zv,$$ 
these terms are taken to be $\beta$-equal because there is a proof $p_1$ determined by a finite sequence of $\beta$-contractions ($\vartriangleright_{1\beta}$) or inverse $\beta$-contractions ($\vartriangleleft_{1\beta}$), possibly with $\alpha$-conversions, which allows for connecting the terms $\lambda x.(\lambda y.yx)zv$ and $zv$, hence 
$$(\lambda x.(\lambda y.yx)z)v\vartriangleright_{1\beta}(\lambda y.yv)z\vartriangleright_{1\beta}zv.$$

Now, the problem is to build a topological model, such that the interpretations of $\lambda$-terms $(\lambda x.(\lambda y.yx)z)v$ and $zv$  are different points, and the proof $p_1$ is a continuous path which connects both
points. This will be used to establish when two proofs (two continuous paths) of a $\beta$-equality between different terms (different points) are ``equal" (homotopic). Hence, in the example, given an second proof $p_2$ which correspond to finite sequence 
$$(\lambda x.(\lambda y.yx)z)v\vartriangleright_{1\beta}(\lambda x.zx)v\vartriangleright_{1\beta}zv,$$
one has that $p_1$ and $p_2$ are two different proofs, so in the model these interpretations should be two different continuous paths. But, would these proof interpretations be homotopically equal? If in the $\lambda$-calculus we call $\beta$-\textit{homotopy} any homotopy of the model,  when are two different $\beta$-homotopies to be declared ``equal"? This can be iterated, and by answering these questions, we could define in the $\lambda$-calculus a  theory of higher $\beta$-equality, with the help of  higher homotopies in the $\lambda$-model. 

\medskip Therefore, this aforementioned theory of higher $\beta$-equality has a structure of a non-trivial $\infty$-groupoid, which extends $\lambda$-calculus to a type-free version of the Homotopy Type Theory (HoTT) \cite{Martinez22}, but with equality relations based on (type-free) computational paths\footnote{If $a,b$ are terms of type $A$, a computational path $s$ from $a$ to $b$ is a composition of rewrites (each rewrite is an application of the inference rules of the equality theory of Martin-L\"of's type theory). One denotes that by $a=_sb$ (see \cite{Queiroz} and \cite{Ramos2017}).}.  Whose advantage is that the $\beta\eta$-conversions  are not equalities of judgment ($a= b:A$), as  in HoTT, but those are intentional equalities ($a=_sb:A$), which could better preserve the information than HoTT does.

\medskip The initiative to search for $\lambda$-models with a $\infty$-groupoid structure emerged in \cite{Martinez} (called \textit{homotopic $\lambda$-models}), which studied the geometry of any complete partial order (c.p.o) (e.g., $D_\infty$), and found that the topology inherent in these models generated trivial higher-order groups. From that moment on, the need arose to look for a
type of model that could present a rich geometric structure, where their higher-order fundamental
groups would not collapse. In this sense, we will gain the semantics of a type-free theory from a version of HoTT based on computational paths, which can distinguish different proofs of equality of $\lambda$-terms.

\medskip According to Quillen's Theorem, each CW complex topological space is homotopically equivalent to a Kan complex\footnote{To ensure the consistency of HoTT, Voevodsky \cite{Voevodsky} (see \cite{Lumsdaine} for higher inductive types) proved that Homotopy Type Theory (HoTT) has a model in the category of Kan complexes. (See \cite{DBLP:books/mk/Univalent13}, p.11)} ($\infty$-groupoid), and, conversely, each Kan complex is homotopically equivalent to a CW complex. Then, instead of working directly with topological spaces, we are going to work with Kan complexes, which are $\infty$-categories  \cite{DBLP:books/mk/Lurie} whose 1-simplexes or edges are weakly invertible. Or, in other words:
  
\begin{defin}[\cite{DBLP:books/mk/Lurie}]
	A  simplicial set $K$ is a Kan complex if for any $0\leq i\leq n$, any map $f_0:\Lambda_i^n\rightarrow K$ admits an extension $f:\Delta^n\rightarrow K$.
\end{defin}

Where the simplicial set $K$ is defined as a presheaf $\Delta^{op}\rightarrow Set$, with $\Delta$ being the \textit{simplicial indexing category}, whose objects are finite ordinals $[n]=\{0,1,\ldots,n\}$, and morphisms are the (non strictly) order preserving maps. $\Delta^n$ is the \textit{standard $n$-simplex} defined for each $n\geq 0$ as the simplicial set $\Delta^n:=\Delta(-,[n])$. And $\Lambda_i^n$ is a \textit{horn} defined as largest subobject of $\Delta^n$ that
does not include the face opposing the $i$-th vertex (see Section \ref{sub:simplicial-sets}).

\medskip Finally, to find Kan complexes that model  $\lambda$-calculus, the strategy would be to generalize the procedure used in \cite{Hyland10}, where to show a way to find categories that  model $\lambda$-calculus, through  the possible solution of domain equations, which are posed on a bicategory with desirable properties of cartesian closure and enough points.

\medskip But, before proposing an $\infty$-category with the properties of cartesian closure and having enough points, we first explore in Section 2 some consequences of the homotopic $\lambda$-models introduced in \cite{Martinez}. In  Section 3, we define the homotopy $\lambda$-model on a cartesian closed $\infty$-category. In Section 4, we adopt the notion of Kleisli structure to the case of the $\infty$-categories and we define the Kleisli $\infty$-category  of a structure. And, finally, in Section 5, we propose an $\infty$-category and we prove that it is closed cartesian and has enough points.

\subsection{Simplicial sets}
\label{sub:simplicial-sets}

For a better understanding of the definitions and basic results on $\infty$-categories which are necessary for the development of this work, we present some    notions on simplicial sets \cite{DBLP:books/mk/Goerss09}.

\begin{defin}[Simplicial indexing category]
	Let $\Delta$ be the category as follows. The objects  are finite ordinals $[n]=\{0,1,\ldots,n\}$,  $n\geq 0$, and  morphisms are the (non strictly) order preserving maps. Morphisms in $\Delta$ are often called simplicial operators.
\end{defin}

\begin{rem}
	There are a coface operator $d^i: [n-1]\rightarrow [n]$, which skips the i-th element and a codegeneracy operator $s^i: [n+1]\rightarrow [n]$, which maps $i$ and $i+1$ to the same element. All operator $f^{\ast}$ in $\Delta$ can be obtained as a cocomposition of coface and codegeneracy operators.
\end{rem}

\begin{defin}[Simplicial set]
	A simplicial set $X$ is a functor $X : \Delta^{op}\rightarrow Set$ (or presheaf). A simplicial
	morphism is just a natural transformation of functors. The category of the simplicial sets $Fun(\Delta^{op},Set)$ will be denoted by $sSet$ or $Set_\Delta$.
\end{defin}

It is typical to write $X_n$ for $X([n])$, and call it the set of $n$-simplexes in $X$.

\begin{rem}
	Given a simplex $a\in X_n$ and a simplicial operator $f^{\ast}: [m]\rightarrow [n]$,  the function $f:X_n\rightarrow X_m$ is given by $f(a):=X(f^{\ast})(a)$. In this explicit language, a simplicial set consists of
	\begin{itemize}
		\item a sequence of sets $X_ 0 , X _1 , X_2,\ldots$,
		\item functions $f: X_n\rightarrow X_m$ for each simplicial operator $f^{\ast}: [m]\rightarrow [n]$. 
	\end{itemize}
	For the coface operator $d^i:[n-1]\rightarrow [n]$, the face map is denoted by $d_i:X_n\rightarrow X_{n+1}$, $0\leq i\leq n$. For the codegeneracy operator $s^i:[n+1]\rightarrow [n]$, the degeneracy map is written by $s_i:X_n\rightarrow X_{n+1}$, $0\leq i\leq n$.
\end{rem} 

\begin{defin}[Product of simplicial sets \cite{Friedman2012}]
	Let $X$ and $Y$ be simplicial sets. Their product $X\times Y$ is defined by
	\begin{enumerate}
		\item $(X\times Y)_n=X_n\times Y_n=\{(x,y)\,|\,x\in X_n,\,y\in Y_n \}$,
		\item if $(x,y)\in(X\times Y)_n$, then $d_i (x,y) = (d_i x, d_i y)$, 
		\item if $(x,y)\in(X\times Y)_n$, then $s_i (x,y) = (s_i x, s_i y)$. 
	\end{enumerate}
\end{defin}

Notice that there are evident projection maps $\pi_1 : X \times Y\rightarrow X$ and $\pi_2 : X\times Y\rightarrow Y$ given by $\pi_1 (x, y) = x$ and $\pi_2(x, y) = y$. These maps are clearly simplicial morphisms.

\begin{defin}[Standard $n$-simplex]
	The standard $n$-simplex $\Delta^n$ is the simplicial set defined by
	$$\Delta^n:=\Delta(-,[n]).$$
	That is, the standard $n$-simplex is exactly the functor represented by the object $[n]$.
\end{defin}

The standard 0-simplex $\Delta^0$ is the
terminal object in $sSet$; i.e., for every simplicial set $X$ there is a unique map $X\rightarrow\Delta^0$. Sometimes we
write $\ast$ instead of $\Delta^0$ for this object. The empty simplicial set $\emptyset$ is the functor $\Delta^{op}\rightarrow Set$ sending each $[n]$ to the empty set. It is the initial object in $sSet$, i.e., for every simplicial set $X$ there is a unique map $\emptyset\rightarrow X$. Besides, there is a bijection $sSet(\Delta^n,X)\cong X_n$; applying the Yoneda Lemma to category $\Delta$ \cite{DBLP:books/mk/Cisinski}.

\medskip A graphical representation of the convex hull of $\Delta^n$ is made up of the by $n+1$ vertices $0,1,\ldots,n$ and the faces are the injective simplicial operators $d^i:[n-1]\rightarrow [n]$, which are called non-degenerated. As seen in the Figure \ref{n-simplexes} for the first four dimensions.

\begin{figure}[ht!]
\centering

\caption{\textmd{Standard $n$-simplexes from $n=0$ to $n=3$}}
\label{n-simplexes}
\fcolorbox{gray}{white}{\includegraphics[width=0.7\textwidth]{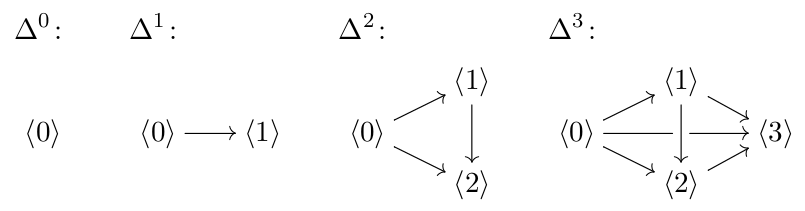}}

\par\selectfont\textbf{{\footnotesize Sink:}} \cite{DBLP:books/mk/Rezk17}  \par
\end{figure} 

\begin{defin}[\cite{DBLP:books/mk/Goerss09}]
	For two simplicial sets $X$, $Y$ we have a mapping  simplicial set, $Map(X,Y)$ defined as:
	$$Map(X,Y)_n := sSet(X \times \Delta^n,Y).$$	
\end{defin}

Note that in particular $Map(X, Y )_0 = sSet(X\times \Delta^0 ,Y)\cong sSet(X, Y)$ (bijection of sets). Sometimes to simplify notation, the simplicial set $Map(X,Y)$ will be written as $X^Y$ or $[X\rightarrow Y]$.

\medskip Next,  a collection of subobjects of the standard simplexes, called “horns” is defined.

\begin{defin}[Horns]
	For each $n \geq 1$, there are subcomplexes $\Lambda_i^n\subset\Delta^n$  for each $0\leq i\leq n$. The horn  $\Lambda_i^n$ is the
	subcomplex of $\Delta^n$ such that this is the largest subobject that does
	not include the face opposing the $i$-th vertex.
	
	\medskip When $0<i<n$ one says that $\Lambda_i^n\subset\Delta^n$ 
	is an inner horn. One also says that it is a left horn if $i<n$ and a right horn if $0<i$.
\end{defin}	

For example, the horns inside $\Delta^1$ are just the vertices: the left horn, the right horn $\Lambda_0^1=\{0\}\subset\Delta^1$ and $\Lambda_1^1=\{1\}\subset\Delta^1$. Neither is an inner horn.  

\bigskip Other example. $\Delta^2$ have three horns: The left horn $\Lambda_0^2$, the internal horn $\Lambda_1^2$ and the right horn $\Lambda_1^2$,  see Figure \ref{Horns}.

\begin{figure}[ht!]
\centering

\caption{\textmd{Horns inside $\Delta^2$}}
\label{Horns}
\fcolorbox{gray}{white}{\includegraphics[width=0.6\textwidth]{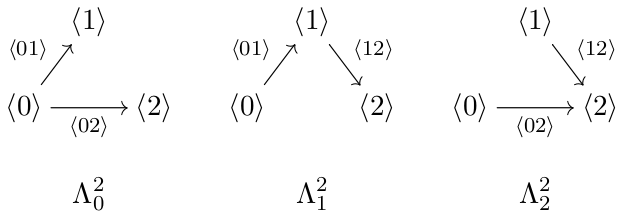}}

\par\selectfont\textbf{{\footnotesize Sink:}} \cite{DBLP:books/mk/Rezk17}  \par
\end{figure}

\subsection{Definition of $\infty$-category and Kan complex}
\label{sub:Definiton-infty-categoy-Kan-complex}

\begin{defin}[$\infty$-category \cite{DBLP:books/mk/Lurie}]
	An $\infty$-category is a simplicial set $X$ which has the following property: for any $0<i<n$, any map $f_0:\Lambda_i^n\rightarrow X$ admits an
	extension $f:\Delta^n\rightarrow X$.	
\end{defin}

\begin{defin}
	From the definition above, we have the following special cases:
	\begin{itemize}
		\item $X$ is a Kan complex if there is extension for each $0\leq i\leq n$.
		\item $X$ is a category if the extension exists uniquely.
		\item $X$ is a groupoid if the extension exists for all $0\leq i\leq n$ and is unique. 
	\end{itemize}
\end{defin}

Next is the definition of cartesian product of $\infty$-categories, which generalizes the cartesian product of categories.

\begin{defin}[Cartesian product]
	A product of $\infty$-categories is the product of the underlying simplicial sets. Thus, $(X\times Y )_n = X_n\times Y_n$ for each $n\geq 0$.
\end{defin}

By the bijective correspondence between the set $sSet(K,X\times Y)$ and $sSet(K\rightarrow X,K\rightarrow Y)$ one has the following proposition (see \cite{DBLP:books/mk/Rezk17}).

\begin{prop}
	The product of two $\infty$-categories (as simplicial sets) is an $\infty$-category.
\end{prop}

\subsection{Categorical constructions in $\infty$-categories}

Next, one approaches the $\infty$-categories from the basic notions of the classical categories.

\begin{defin}
	A functor of $\infty$-categories $X\rightarrow Y$ is exactly a morphism of simplicial sets. Thus, $Fun(X,Y)=Map(X,Y)$ must be a simplicial set of the functors from $X$ to $Y$. 
\end{defin}

\begin{nota}
	The notation $Map(X,Y)$  is normally used for simplicial sets, while $Fun(X,Y)$ is for $\infty$-categories. One will refer to morphisms in $Fun(X, Y)$ as natural transformations of functors,
	and equivalences in $Fun(X, Y)$ as natural equivalences. 
	
	\medskip The composition of $n$-simplex $f:\Delta^n\rightarrow X$ (or $f\in X_n$) with a functor $F:X\rightarrow Y$, will be denoted as the image $F(f)\in Y_n$, where $n\geq 0$.
	
	\medskip A 0-simplex or vertex $x:\Delta^0\rightarrow X$ will be denoted as an object $x\in X$ in the $\infty$-category $X$.
	
	\medskip A 1-simplex $f:\Delta^1\rightarrow X$, such that $f(0)=x$ and $f(1)=y$ will be denoted as a morphism $f:x\rightarrow y$ in the $\infty$-category $X$.
	
	\medskip An inner horn $\Lambda_1^2\rightarrow X$, which corresponds to composable morphisms $x\xrightarrow{f}y\xrightarrow{g}z$
	in the $\infty$-category $X$, will be denoted by $(g,-,f)$ or in some cases to simplify notation it will be denoted by $g.f$.
\end{nota}

\begin{prop}[\cite{DBLP:books/mk/Lurie} and \cite{DBLP:books/mk/Cisinski}]
	For every $\infty$-category $Y$, the simplicial set $Fun(X, Y)$ is an $\infty$-category.
\end{prop}

The proof of the following theorem can be found in \cite{DBLP:books/mk/Lurie} and \cite{DBLP:books/mk/Cisinski}.

\begin{teor}[Joyal]
	A simplicial set $X$ is an $\infty$-category if and only if the canonical morphism 
	$$Fun(\Delta^2,X)\rightarrow Fun(\Lambda_1^2,X),$$
	is a trivial fibration. Thus, each fibre of this morphism is contractible.
\end{teor}

The above theorem guarantees the laws of coherence of the composition of 1-simplexes or morphisms of an $\infty$-category. This means that the composition of morphisms is unique up to homotopy, i.e., the composition is well-defined up to a space of choices is contractible (equivalent to $\Delta^0$).

\medskip With respect to the Kan complexes, we have the following equivalence.

\begin{prop}[Homotopy extension lifting property \cite{DBLP:books/mk/Lurie}]
	The simplicial set $X$ is a Kan complex if and only if the induced map	
	$$Fun(\Delta^1,X)\rightarrow Fun(\{0\},X)$$
	is a trivial fibration of simplicial sets.
\end{prop}

\begin{defin}[Space of morphisms \cite{DBLP:books/mk/Rezk17}]
	For two vertices $x,y$ in an $\infty$-category $X$, define the space of morphisms $X(x,y)$ by the following pullback diagram:	
	\[\xymatrix{
		& {X(x,y)}\ar[d]_{}\ar[r]^{} &
		Fun(\Delta^1,X) \ar[d]^{(s,t)}\\
		& \Delta^0\ar[r]_{(x,y)} & X\times X
		& 
	}\]
\end{defin}

\begin{prop}[\cite{DBLP:books/mk/Lurie} and \cite{DBLP:books/mk/Rezk17}]
	The morphism spaces $X(x,y)$ are Kan complexes.
\end{prop}

\subsection{Equivalences in $\infty$-categories}

In category theory, we have the concept of isomorphism of objects. For the case of the $\infty$-categories, we will have the equivalence of objects (vertices) in the following sense.  

\begin{defin}[Equivalent vertices]
	A morphism (1-simplex) $f:x\rightarrow y$ in an $\infty$-category $X$ is invertible (an equivalence) if there is a morphism $g:y\rightarrow x$ in $X$, a pair of 2-simplexes $\alpha,\beta\in X_2$ such that 
	$(g,-,f)\xrightarrow{\alpha}1_x$ and $(f,-,g)\xrightarrow{\beta}1_y$, i.e., if the diagram
	\[\xymatrix{
		& {x}\ar[dd]_{1_x}\ar[rr]^f& &
		y \ar[dd]^{1_y}\ar@{-->}[lldd]^g\\ \\
		& x\ar[rr]_{f}& & y
		& 
	}\]
	commutes under the 2-simplexes $\alpha$ and $\beta$.	
\end{defin}

\begin{teor}[\cite{DBLP:books/mk/Lurie} and \cite{DBLP:books/mk/Cisinski}]
	Let $X$ be an $\infty$-category. The following are equivalent
	\begin{enumerate}
		\item Every morphism (1-simplex) in $X$ is an equivalence.
		\item  $X$ is a Kan complex.
	\end{enumerate}
\end{teor}

\subsection{Natural transformations and natural equivalence}

\begin{defin}[\cite{DBLP:books/mk/Cisinski} and \cite{DBLP:books/mk/Rezk17}]
	If $X$ and $Y$ are $\infty$-categories, and if $F, G : X\rightarrow Y$ are two functors, a natural
	transformation from $F$ to $G$ is a map $H : X\times \Delta^1\rightarrow Y$ such that
	$$H(x,0)=F(x),\hspace{0.5cm}H(x,1)=G(x),$$
	for each vertex $x\in X$. Such a natural transformation is invertible or it is a natural equivalence if for any vertex $x\in X$, the
	induced morphism $F(x)\rightarrow G(x)$ (corresponding to the restriction of $H$ to
	$\Delta^1 \cong \{x\}\times\Delta^1$) is invertible in $Y$. If there is a natural equivalence from $F$ to $G$, we write $F\simeq G$.
\end{defin}

\begin{rem}
	This means that for each vertex $x\in X$, one chooses a morphism $H_x:F(x)\rightarrow G(x)$ such that the following diagram 
	\[\xymatrix{
		& {F(x)}\ar[dd]_{F(f)}\ar[rr]^{H_x}\ar@{-->}[rrdd]^g& &
		G(x) \ar[dd]^{G(f)}\\ \\
		& F(x')\ar[rr]_{H_{x'}}& & G(x')
		& 
	}\]
	commutes under the 2-simplexes $\alpha:g\rightarrow(G(f),-,H_x)$ and $\beta:(H_{x'},-,F(f))\rightarrow g$.
\end{rem}

\subsection{Categorical equivalences and homotopy equivalences}

\begin{defin}[Categorical equivalence \cite{DBLP:books/mk/Rezk17} and \cite{DBLP:books/mk/Lurie}]
	A functor of $\infty$-categories $F:X\rightarrow Y$ is a categorical equivalence if there is another functor $G:Y\rightarrow X$, such that $GF\simeq 1_X$ and $FG\simeq 1_Y$.	
\end{defin}

\begin{rem}
	From the definition above, if  $F:X\rightarrow Y$ is a functor of Kan complexes, we say that $F$ is a homotopy equivalence.
\end{rem}

\begin{lem}[\cite{DBLP:books/mk/Rezk17} and \cite{DBLP:books/mk/Cisinski}]
	A functor of $\infty$-categories $F:X\rightarrow Y$ is a categorical equivalence
	if it satisfies the following two conditions:
	\begin{itemize}
		\item Fully Faithful (Embedding): For two objects $x,y\in X$ the induced functor of Kan complexes
		$$X(x,y)\rightarrow Y(Fx,Fy),$$
		is a homotopy equivalence.
		\item Essentially Surjective: For every object $y\in Y$ there exists an object $x\in Y$ such that $Fx$ is equivalent to $y$.
	\end{itemize}
\end{lem} 

\subsection{The join of $\infty$-categories}
\label{sub:section3}

Next, the extension from join of categories to $\infty$-categories. This will enable us to define limit and colimit in an $\infty$-category.

\begin{defin}[Join \cite{DBLP:books/mk/Lurie}]
	Let $K$ and $L$ be simplicial sets. The join $K\star L$ is the simplicial set defined by
	$$(K\star L)_n:=K_n\cup L_n\cup\bigcup_{i+1+j=n}K_i\times L_j, \hspace{0.5cm}n\geq 0.$$ 
\end{defin}

\begin{ejem}[\cite{DBLP:books/mk/Groth15}]
	\begin{enumerate}
		\item If $K\in sSet$ and $L=\Delta^0$, then $K^\vartriangleright=K\star\Delta^0$ is the cocone or the right cone on $K$. Dually,  If $L\in sSet$ then $L^\vartriangleleft=\Delta^0\star L$ is the cone or the left cone on $L$.
		\item Let $K=\Lambda_0^2$. If we see this left horn as a pushout, the cocone $(\Lambda_0^2)^\vartriangleright$ is isomorphic to the square $\square=\Delta^1\times\Delta^1$, that is, to the filled in diagram
		\[\xymatrix{
			& {(0,0)}\ar[dd]_{}\ar[rr]^{}\ar[rrdd]& &
			(1,0) \ar[dd]^{}\\ \\
			& (0,1)\ar[rr]_{}& & (1,1)
			& 
		}\] 
	\end{enumerate}
\end{ejem}

\begin{prop}[\cite{DBLP:books/mk/Lurie}]
	
	\begin{enumerate}
		\item[(i)] For the standard simplexes one has an isomorphism $\Delta^i\star\Delta^j\cong \Delta^{i+1+j}$, $i,j\geq 0$, and these isomorphisms are with the obvious inclusions of $\Delta^i$ and $\Delta^j$.
		\item[(ii)] If $X$ and $Y$ are $\infty$-categories, then the join $X\star Y$ is an $\infty$-category.
	\end{enumerate}
\end{prop}

\subsection{The slice $\infty$-category}

In the case of the classical categories, if $A,B$ are categories and $p:A\rightarrow B$ is any functor, one can form the slice category $B_{/p}$ of the object over $p$ or cones on $p$. The following propositions allow us to define the slice $\infty$-category.   

\begin{prop}[\cite{Joyal2002}]
	Let $K$ and $S$ be simplicial sets, and $p:K\rightarrow S$ be an arbitrary map. There is a simplicial set $S_{/p}$ such that there exists a natural bijection 
	$$sSet(Y,S_{/p})\cong sSet_p(Y\ast K,S),$$
	where the subscript on the righthand side indicates that we consider only those morphisms $f:Y\ast K\rightarrow S$ such that $f|K=p$.
\end{prop}

\begin{prop}[Joyal]
	Let $X$ be an $\infty$-category and $K$ be a simplicial set. If $p:K\rightarrow X$ is a map of simplicial sets, then $X_{/p}$ is an $\infty$-category. Moreover, if $q:X\rightarrow Y$ is a categorical equivalence, then the induced map $X_{/p}\rightarrow Y_{qp}$ is a categorical equivalence as well.  
\end{prop}

\begin{defin}[Slice $\infty$-Categorical \cite{DBLP:books/mk/Lurie}]
	Let $X$ be an $\infty$-category, $K$ be a simplicial set and $p:K\rightarrow X$ be a map of simplicial sets. Define the slice $\infty$-category $X_{/p}$ of the objects over $p$ or cones on $p$. Dually, $X_{p/}$ is the $\infty$-category of objects under $p$ or cocones on $p$.
\end{defin}

\begin{ejem}
	Let $X$ be an $\infty$-category and $x\in X$ be an object, which corresponds to map  $x:\Delta^0\rightarrow X$. The objects of the $\infty$-category $X_{/x}$ of cones on $x$ are morphisms $y\rightarrow x$ in $X$, and the morphism from $y\rightarrow x$ to $z\rightarrow x$ in $X_{/x}$, are the 2-simplexes 
	\[\xymatrix{
		& {y}\ar[dd]_{}\ar[rrdd]& &
		\\ \\
		& z\ar[rr]_{}& & x
		& 
	}\]
	in the $\infty$-category $X$. 
\end{ejem}

\subsection{Limits and colimits}

An object $t$ of a category is final if for each object $x$ in this category, there is a unique morphism $x\rightarrow t$. Next, one defines the final objects at $\infty$-categories, under a contractible space of morphisms.

\begin{defin}[Final object \cite{DBLP:books/mk/Lurie}]
	An object $\omega\in X$ in an $\infty$-category $X$ is a final object if for any object $x\in X$, the Kan complex of morphisms $X(x,\omega)$ is contractible.
\end{defin}

\begin{teor}[Joyal]
	Take an object $\omega\in X$ in an $\infty$-category $X$ and $\pi:X_{/\omega}\rightarrow X$ the canonical projection. The following conditions are equivalent.
	\begin{enumerate}
		\item[(i)] The object $\omega\in X$ is final. 
		\item[(ii)] The map $\pi:X_{/\omega}\rightarrow X$ is a trivial fibration.
		\item[(iii)] The map $\pi:X_{/\omega}\rightarrow X$ is a categorical equivalence.
		\item[(iv)] The map  $\pi:X_{/\omega}\rightarrow X$ has a section which sends $\omega$ to $1_\omega$.
		\item[(v)] Any map $f_0:\partial\Delta^n\rightarrow X$, such that $n>0$ and $f(n)=\omega$, has an extension $f:\Delta^n\rightarrow X$.   
	\end{enumerate}
\end{teor}

\begin{corol}[\cite{DBLP:books/mk/Cisinski}]\label{corolally-final-objects}
	The final objects of an $\infty$-category $X$ form a Kan complex which is either empty or equivalent to the point. 
\end{corol}

\begin{corol}[\cite{DBLP:books/mk/Cisinski}]
	Let $x$ be a final object in an $\infty$-category $X$. For any simplicial set $A$, the constant map $A\rightarrow X$ with value $x$ is a final object in $Map(A,X)$.  
\end{corol}

\begin{defin}[Limit and colimit \cite{Joyal2002}]
	Let $X$ be an $\infty$-category and let $p:K\rightarrow X$ be a map of simplicial sets. A colimit for $p$ is an initial object of $X_{p/}$, and a limit for $p$ is a final object of $X_{/p}$. 
\end{defin}

By the dual of Corollary \ref{corolally-final-objects}, if the colimit exists, then the Kan complex of initial objects is contractible, i.e., the initial object is unique up to contractible choice. 

\subsection{$\infty$-categories of presheaves}   

\begin{defin}[$\infty$-categories of presheaves \cite{DBLP:books/mk/Lurie}]
	Let $S$ be a simplicial set. One lets $P(S)$ or $PS$ denote simplicial set $Fun(S^{op},\mathscr{S})$, where $\mathscr{S}$ denotes the $\infty$-category of the small Kan complexes or $\infty$-groupoids, also called the $\infty$-category of the spaces. One will say that $P(S)$ is the $\infty$-category of the presheaves on $S$.
\end{defin}

\begin{prop}[\cite{DBLP:books/mk/Lurie}]
	Let $S$ be a simplicial set. The $\infty$-category $P(S)$ of the presheaves on $S$ admits all small limits and colimits.
\end{prop}

\begin{prop}[$\infty$-Categorical Yoneda Lemma \cite{DBLP:books/mk/Lurie}]
	Let $S$ be a simplicial set. Then the Yoneda embedding $j:S\rightarrow PS$ is fully faithful. 
\end{prop}

\begin{nota}
	Let $X$ be an $\infty$-category and $S$ be a simplicial set. One lets $Fun^L(PS,X)$ denote the full subcategory of $Fun(PS,X)$ spanning those functors $PS\rightarrow X$ which preserve small colimits. 
	
	\medskip The motivation for this notation stems from Adjoint Functor Theorem (will be seen later), where $Fun^L(PS,X)$ also denotes the full subcategory of $Fun(PS,X)$ spanning those functors  which are left adjoints. 
\end{nota}

\begin{teor}[\cite{DBLP:books/mk/Lurie}]
	Let $S$ be a small simplicial set and let $X$ be an $\infty$-category which admits small colimits. The composition with the Yoneda embedding $j:S\rightarrow PS$ induces an equivalence of $\infty$-categories
	$$Fun^L(PS,X)\rightarrow Fun(S,X).$$
	
\end{teor}

\subsection{Some $\infty$-categories of presheaves}
In the literature, such as can be seen in \cite{DBLP:books/mk/Lurie} and \cite{DBLP:books/mk/Cisinski}, one finds the definition of the $\infty$-category of presheaves on a small $\infty$-category $B$ as $PB=[B^{op},\mathscr{S}]$ which sets the categorical equivalence 
$$Fun(A,PB)\simeq Fun(A\times B^{op},\mathscr{S}),$$
where $\mathscr{S}$ is the $\infty$-category of all small Kan complexes. The $\infty$-category $PB$ and this equivalence are defined on the $\infty$-category of all $\infty$-categories $CAT_\infty$. 

\medskip Another fundamental result is the following: given a collection of simplicial sets $\mathscr{K}$,  $\mathscr{R}\subseteq\mathscr{K}$  and  $A$ an $\infty$-category, there exists an $\infty$-category $P^\mathscr{K}_\mathscr{R}A$ and a functor $j:A\rightarrow P^\mathscr{K}_\mathscr{R}A$ with the following properties:
\begin{enumerate}
	\item $P^\mathscr{K}_\mathscr{R}A$ admits $\mathscr{K}$-indexed colimits, i.e., admits $K$-indexed colimits for each $K\in\mathscr{K}$.
	\item For every $\infty$-category $B$ which admits $\mathscr{K}$-indexed colimits, composition with $j$ induces an equivalence of $\infty$-categories 
	$$Fun_\mathscr{K}(P^\mathscr{K}_\mathscr{R}A,B)\simeq Fun_\mathscr{R}(A,B).$$
	If $A$ admits all the $\mathscr{R}$-indexed colimits, we also have
	\item The functor $j$ is fully faithful. 
\end{enumerate}
where $Fun_\mathscr{R}(A,B)$ is the full subcategory of $Fun(A,B)$ spanned by  those functors which preserve $\mathscr{R}$-indexed colimits, i.e., which preserve $K$-indexed colimits for each $K\in\mathscr{R}$; the same applies to  $Fun_\mathscr{K}(P^\mathscr{K}_\mathscr{R}A,B)$.

\begin{ejem}
	Let $\mathscr{R}=\emptyset$ and $\mathscr{K}$ be the class of all small simplicial sets. If $A$ is a small $\infty$-category, then $P^\mathscr{K}_\mathscr{R}A\simeq PA$.
\end{ejem}

\begin{ejem}
	Let $\mathscr{R}=\emptyset$ and $\mathscr{K}$ be the class of all small $\kappa$-filtered simplicial sets for some regular cardinal $\kappa$. If $A$ is a small $\infty$-category, then $P^\mathscr{K}_\mathscr{R}A\simeq Ind_\kappa A$.
\end{ejem}

\begin{ejem}\label{k-small-limits-example}
	Let $\mathscr{R}=\emptyset$ and $\mathscr{K}$ be the class of all $\kappa$-small simplicial sets for some regular cardinal $\kappa$. If $A$ is a small $\infty$-category, then $P^\mathscr{K}_\mathscr{R}A\simeq P^\kappa A$, where $P^\kappa A$ is the full subcategory of all $\kappa$-compact  elements of $PA$. 
\end{ejem}

\begin{ejem}
	Let $\mathscr{R}$ be the class of all $\kappa$-small simplicial sets for some regular cardinal $\kappa$ and let $\mathscr{K}$ be the collection of all small simplicial sets. Let $A$ be a small $\infty$-category which admits $\kappa$-small colimits, then  $P^\mathscr{K}_\mathscr{R}A\simeq Ind_\kappa A$. Also, we have $A\simeq P^\kappa C$ for some small $\infty$-category $C$ which does not necessarily admit $\kappa$-small colimits.
\end{ejem}

\section{Arbitrary syntactical homotopic $\lambda$-models}

In this section, we discuss some consequences of the arbitrary syntactic homotopic lambda models introduced in \cite{Martinez}, which correspond to a direct generalization (2-dimensional) of the traditional structured set models of a cartesian closed category (1-dimensional) as can be seen in \cite{DBLP:books/mk/Barendregt} and \cite{DBLP:books/mk/HindleyS08}.

\begin{nota}

For $K$ being a Kan simplex and $n\geq 0$, let $K_n=Fun(\Delta^n,K)$ be the Kan complex of the $n$-simplexes.

Denote by $\Omega_n(K,a)$ the class of all the spheres $\partial\Delta^n\rightarrow K$ with initial vertex  $a\in K$. 

Let $Var$ be the set of all variables of $\lambda$-calculus, for all $n\geq 0$, each assignment  $\rho:Var\rightarrow K_n$ ($\rho(t)$ is an $n$-simplex of $K$, for each $t\in Var$), $x\in Var$ and $f\in K_m$. Denote by $[f/x]\rho:Var\rightarrow K$ the assignment
\begin{equation*}
([f/x]\rho)(t)=
\begin{cases}
f & \text{if \, $t=x$}\\
\rho (t) & \text{if \, $t\neq x$.}
\end{cases}
\end{equation*}
\end{nota}

\begin{defin}[Syntactic Homotopic $\lambda$-model]\label{homotopic-lambda-model-definition}
	A  homotopic $\lambda$-model is a triple $\langle K,\bullet, \llbracket \, \rrbracket\rangle$, where $K$ is a Kan complex, $\bullet:K\times K\rightarrow K$ is a functor,  and $\llbracket \, \rrbracket$ is a mapping which assigns to $\lambda$-term $M$ and each assignment $\rho:Var\rightarrow K_n$, an $n$-simplex $\llbracket P \rrbracket_{\rho}$ in $K$ for each $n\geq 0 $ such that
	\begin{enumerate}
		\item $\llbracket x \rrbracket_{\rho}=\rho(x);$
		\item $\llbracket MN \rrbracket_{\rho}=\llbracket M\rrbracket_{\rho}\bullet\llbracket N\rrbracket_{\rho};$
		\item For each $f\in K_{n}$, there is a limit $\beta_{f}:\llbracket\lambda x.M\rrbracket_{\rho}\bullet f\rightarrow\llbracket M\rrbracket_{[f/x]\rho}$ from $\llbracket M\rrbracket_{[f/x]\rho}\in K_n$;
		\item $\llbracket M\rrbracket_{\rho}=\llbracket M\rrbracket_{\sigma}$ if $\rho(x)=\sigma(x)$ for $x\in FV(M)$;
		\item $\llbracket\lambda x.M(x)\rrbracket_{\rho}=\llbracket\lambda y.M(y)\rrbracket_{\rho}$ if $y\notin FV(M)$;
		\item if $(\forall a\in K_0)(\forall n\geq 1)(\forall \omega\in \Omega_n(K,a))\left( \llbracket M\rrbracket_{[\omega/x]\rho}=\llbracket N\rrbracket_{[\omega/x]\rho}\right) $, then 
		
		$\llbracket\lambda x.M\rrbracket_{\rho}=\llbracket\lambda x.N\rrbracket_{\rho}$.
	\end{enumerate}
	The homotopic model $\langle K,\bullet, \llbracket \, \rrbracket\rangle$ is an extensional syntactic homotopic model if it satisfies the additional property: there is a colimit $\eta:\llbracket M\rrbracket_{\rho}\rightarrow\llbracket\lambda x.Mx\rrbracket_{\rho}$ from  $\llbracket M\rrbracket_{\rho}\in K_n$ with $x\notin FV(M)$.
\end{defin}

\begin{rem}
	Note that the condition (3) of the Definition \ref{homotopic-lambda-model-definition}, by the Homotopy Extension Lifting Property \cite{DBLP:books/mk/Lurie}, if $x\in FV(P)$ any cone $\beta_{f}:\llbracket\lambda x.P\rrbracket_{\rho}\bullet f\rightarrow\llbracket P\rrbracket_{[f/x]\rho}$ in $K_{/\llbracket P\rrbracket_{[f/x]\rho}}$ is a limit of $n$-simplex $\llbracket P\rrbracket_{[f/x]\rho}$. Since $K_n$ is a Kan complex, by the theorem mentioned above, the induced functor  $Fun(\Delta^1,K_n)\rightarrow Fun(\Delta^0,K_n)$ is a trivial fibration, hence the fibre $(K_n)_{/\llbracket P\rrbracket_{[f/x]\rho}}$ is contractible, that is $K_{/\llbracket P\rrbracket_{[f/x]\rho}}$ is contractible. Thus the condition (3) is reduced to the existence of a cone $\beta_{f}:\llbracket\lambda x.P\rrbracket_{\rho}\bullet f\rightarrow\llbracket P\rrbracket_{[f/x]\rho}$ in $(K_n)_{/\llbracket P\rrbracket_{[f/x]\rho}}$. 
\end{rem}

\begin{defin}
	Let $\mathfrak{M}=\langle K,\bullet,\llbracket \, \rrbracket\rangle$ be  a syntactic homotopic $\lambda$-model. The notion of satisfaction in $\mathfrak{M}$ is defined as
	$$\mathfrak{M},\rho\models M=N\hspace{0.2cm}\Longleftrightarrow\hspace{0.2cm}\llbracket M \rrbracket_\rho\simeq\llbracket N \rrbracket_\rho$$
	$$\hspace{1.4cm}\mathfrak{M}\models M=N\hspace{0.2cm}\Longleftrightarrow\hspace{0.2cm}\forall\rho\,(\mathfrak{M},\rho\models M=N)$$
\end{defin}

\begin{lem}
	Let $\mathfrak{M}=\langle K,\bullet,\llbracket \, \rrbracket\rangle$ be a syntactic homotopic $\lambda$-model. Then, for all $M$, $N$, $x$, $n\geq 0$ and $\rho:Var\rightarrow K_n$, 
	\begin{enumerate}
		\item[(i)]  $\llbracket[z/x] M\rrbracket_\rho=\llbracket M\rrbracket_{[\rho(z)/x]\rho}$,
		\item[(ii)] if $\llbracket[N/x] M\rrbracket_\rho=\llbracket M\rrbracket_{[\llbracket N\rrbracket_\rho/x]\rho}$, then 	$\llbracket\lambda y.[N/x] M\rrbracket_\rho=\llbracket\lambda y. M\rrbracket_{[\llbracket N\rrbracket_\rho/x]\rho},$ 
		\item[(iii)] 	$\llbracket[N/x] M\rrbracket_\rho=\llbracket M\rrbracket_{[\llbracket N\rrbracket_\rho/x]\rho}.$  
	\end{enumerate}

\end{lem}
\begin{proof}
	(i) One has that,
		\begin{align*}
	\llbracket [z/x]M\rrbracket_\rho=\llbracket  [z/x]M\rrbracket_{[\rho(z)/z]\rho} 
	&\xleftarrow{\beta_{\rho(z)}} \llbracket\lambda z. [z/x]M\rrbracket_{\rho}\bullet\rho(z)=\llbracket\lambda x. M\rrbracket_{\rho}\bullet\rho(z); \\
	\llbracket  M\rrbracket_{[\rho(z)/x]\rho}&\xleftarrow{\beta_{\rho(z)}}\llbracket\lambda x. M\rrbracket_{\rho}\bullet\rho(z) 
	\end{align*}
	That is, $\llbracket[z/x] M\rrbracket_\rho=\llbracket M\rrbracket_{[\rho(z)/x]\rho}$. 
	
	\medskip (ii) First suppose $x\notin FV(N)$. Let $y\neq x$ and $y\notin FV(N)$. For $\rho'=[\llbracket N\rrbracket_\rho/x]\rho$ and any $\omega\in\Omega_n(K,a)$, with an arbitrary vertex $a\in K$ and any $n\geq 1$, one has
    \begin{align*}
		\llbracket[N/x] M\rrbracket_{[\omega/y]\rho'}&=\llbracket [N/x]M\rrbracket_{[\omega/y]\rho} \\
		&=\llbracket M\rrbracket_{[\omega/y][\llbracket N\rrbracket_\rho/x]\rho};\hspace{0.5cm} \text{by hypothesis} \\
		&=\llbracket M\rrbracket_{\omega/y]\rho'}.
	\end{align*}
	By Definition \ref{homotopic-lambda-model-definition} (6), 
	 $\llbracket\lambda y.[N/x] M\rrbracket_{\rho'}=\llbracket\lambda y.M\rrbracket_{\rho'}$, hence
	 $$\llbracket\lambda y.[N/x] M\rrbracket_{\rho}=\llbracket\lambda y.[N/x] M\rrbracket_{\rho'}=\llbracket\lambda y. M\rrbracket_{[\llbracket N\rrbracket_\rho/x]\rho}.$$
	 
	 If $x\in FV(N)$, the proof is identical to  \cite[p.103]{DBLP:books/mk/Barendregt}.
	  
	 \medskip (iii) Follows easy by induction on the $\lambda$-term $M$.
\end{proof}
\begin{teor}
	Let $\mathfrak{M}=\langle K,\bullet,\llbracket \, \rrbracket\rangle$ be a syntactic homotopic $\lambda$-model. Then
	$$\lambda\beta\vdash M=N\hspace{0.2cm}\Longrightarrow\hspace{0.2cm} \mathfrak{M}\models M=N.$$
\end{teor}
\begin{proof}
	By induction on the length of proof. For the axiom $(\lambda x.M)N=[N/x]M$ we proceed
	\begin{align*}
	\llbracket  (\lambda x.M)N\rrbracket_\rho&=\llbracket  \lambda x.M\rrbracket_\rho\bullet \llbracket N\rrbracket_\rho \\
	&\xrightarrow{\beta_{\llbracket N\rrbracket_\rho }} \llbracket M\rrbracket_{[\llbracket N\rrbracket_\rho/x]\rho} \\
	&= \llbracket[N/x] M\rrbracket_\rho
	\end{align*}
	The rule $M=N\Longrightarrow \lambda x.M=\lambda x.N$ follows from Definition \ref{homotopic-lambda-model-definition} (6). The other rules are trivial.
\end{proof}

\begin{defin}[h.p.o]\label{h.p.o}
	Let $\hat{K}$ be an $\infty$-category. The largest Kan complex $K\subseteq\hat{K}$ is a homotopy partial order (h.p.o), if for every $x,y\in K$ one has that $\hat{K}(x,y)$ is contractible or empty. Hence, the Kan complex $K$ admits a relation of h.p.o $\precsim$ defined for each $x,y\in K$ as follows:
	$x\precsim y$ if $\hat{K}(x,y)\neq\emptyset$, hence the pair $(K,\precsim)$ is a h.p.o. (we denote simply by $K$). The $\infty$-category $\hat{K}$  is also called a h.p.o.
\end{defin}

\begin{defin}[c.h.p.o] 
	Let $K$ be an h.p.o.
	\begin{enumerate}
		\item An h.p.o $X\subseteq K$ is directed if $X\neq \emptyset$ and for each $x,y\in X$, there exists $z\in X$ such that $x\precsim z$ and $y\precsim z$.
		\item $K$ is a complete homotopy partial order (c.h.p.o) if
		\begin{enumerate}
			\item There are initial objects, i.e.,  $\bot\in K$ is a initial object if for each $x\in K$, $\bot\precsim x$. 
			\item For each directed $X\subseteq\mathcal{K}$ the supremum (or colimit) $\bigcurlyvee X\in\mathcal{K}$ exists. 
		\end{enumerate}
	\end{enumerate}
\end{defin}

\begin{defin}[Reflexive and Extensional Kan complex]\label{Definition Reflexive and Extensional Kan complex}
	Let $K$ be a c.h.p.o. The Kan complex $K$ is called reflexive if the full subcategory  $[K\rightarrow K]\subseteq Fun(K,K)$ of the continuous functors is a retract of $K$, i.e., there are continuous functors
	$$F:K\rightarrow[K\rightarrow K],\hspace{1cm} G:[K\rightarrow K]\rightarrow K$$  
	such that there is a natural equivalence $\varepsilon:FG\rightarrow id_{[K\rightarrow K]}$.
	
	\medskip If there is a natural equivalence $\eta:id_K\rightarrow GF$, we call  $K$ an extensional Kan complex. 
\end{defin}

In \cite{Martinez21}, we proved the existence of extensional Kan complexes.

\begin{defin}\label{interpretation-definition}
	Let $K$ be a reflexive Kan complex (via $F$, $G$ and $\varepsilon$).
	\begin{enumerate}
		\item For $f,g:\Delta^n\rightarrow K$ (or also $f,g\in K_n$) define the $n$-simplex
		$$f\bullet_{\Delta^n} g=F(f)(g).$$
		In particular for vertices $a,b\in K$,
		$$a\bullet b=a\bullet_{\Delta^0} b=F(a)(b),$$
		
		besides, $F(a)\bullet(-)=a\bullet(-)$ and $F(-)(b)=(-)\bullet b$ are functors on $K$, then for $f\in K_n$ one defines the $n$-simplexes 
		$$a\bullet f=F(a)(f),\hspace{1cm}f\bullet b=F(f)(b).$$ 
		
		\item For each $n\geq 0$, let $\rho$ be a valuation at $K_n$. Define the interpretation $\llbracket \,\,\, \rrbracket_\rho:\Lambda\rightarrow K_n$ by induction as follows
		\begin{enumerate}
			\item $\llbracket x \rrbracket_\rho=\rho(x),$
			\item $\llbracket MN \rrbracket_\rho=\llbracket M \rrbracket_\rho\bullet\llbracket N \rrbracket_\rho$,
			\item $\llbracket \lambda x.M \rrbracket_\rho= G(\boldsymbol{\lambda} f.\llbracket M \rrbracket_{[f/x]\rho})$, where $\boldsymbol{\lambda} f.\llbracket M \rrbracket_{[f/x]\rho}=\llbracket M \rrbracket_{[-/x]\rho}$. 
		\end{enumerate}
	\end{enumerate}
\end{defin}

\begin{lem}
	If $n\geq 0$ and $\rho:Var\rightarrow K_n$, then  $\boldsymbol{\lambda} f.\llbracket M \rrbracket_{[f/x]\rho}$ defines a functor $\Delta^n\rightarrow [K\rightarrow K]$; hence $\llbracket \lambda x.M \rrbracket_\rho$ is well-defined in Definition \ref{interpretation-definition} (2.c).
	
	\medskip\textit{Proof}.
		By induction on $P$ we show that $\boldsymbol{\lambda} f.\llbracket P \rrbracket_{[f/(x)(i)]\rho}$ defines a functor  $K\times \{i\}\rightarrow K$ for each vertex $i\in\Delta^n$ and all $\rho$ in $K_n$, where the map $\llbracket P \rrbracket_{[-/(x)(-)]\rho}:K\times\Delta^n\rightarrow K$ (with $\rho(x)(-):\Delta^n\rightarrow K$) depicts to $\boldsymbol{\lambda} f.\llbracket P \rrbracket_{[f/x]\rho}:K\rightarrow K_n$. 
		
		For each $f:\Delta^m\rightarrow K$ one has: 
		\begin{enumerate}
			\item[(a)] $\llbracket x \rrbracket_{[f/(x)(i)]\rho}=f\in K_m$. So $\boldsymbol{\lambda} f.\llbracket x \rrbracket_{[f/(x)(i)]\rho}=I_{K}$ (Identity functor) , which is continuous.
			\item[(b)] $\llbracket x \rrbracket_{[f/(y)(i)]\rho}=s^m(\rho(x)(i))\in K_m$, with $s^m$ the degeneration operator applied m-times to vertex $\rho(x)(i)$.  Then $\boldsymbol{\lambda} f.\llbracket x \rrbracket_{[f/(y)(i)]\rho}$ is the constant functor in the vertex $\rho(x)(i)$, which is continuous.
			\item[(c)] $\llbracket MN \rrbracket_{[f/(x)(i)]\rho}= \llbracket M \rrbracket_{[f/(x)(i)]\rho}\bullet_{\Delta^m}\llbracket N \rrbracket_{[f/(x)(i)]\rho}\in K_m$;  since by I.H (Induction Hypothesis)  $\llbracket M \rrbracket_{[f/(x(i))]\rho}$, $\llbracket N \rrbracket_{[f/(x)(i)]\rho}$ are $m$-simplexes (can be degenerates), hence $\llbracket MN \rrbracket_{[f/(x)(i)]\rho}$ is an m-simplex. Besides, the functor $\llbracket MN \rrbracket_{[-/(x)(i)]\rho}=F(\llbracket M \rrbracket_{[-/(x)(i)]\rho})(\llbracket N \rrbracket_{[-/(x)(i)]\rho})$ is continuous by I.H and continuity of $F$. 
			
			\item[(d)] $\llbracket \lambda y.M  \rrbracket_{[f/(x)(i)]\rho}=G(\boldsymbol{\lambda}g.\llbracket M \rrbracket_{[g/y][f/(x)(i)]\rho})\in K_m$;  by I.H the map
			
			 $\boldsymbol{\lambda}f.\boldsymbol{\lambda}g.\llbracket M \rrbracket_{[g/y][f/x(i)]\rho}:K\rightarrow [K\rightarrow K]$ is a continuous functor in $f$ and $g$ separately, so is continuous \cite{Martinez21}. Thus, $\boldsymbol{\lambda}g.\llbracket M \rrbracket_{[g/y][f/(x)(i)]\rho}$ is an $m$-simplex at $[K\rightarrow K]$, applying the continuous functor $G:[K\rightarrow K]\rightarrow K$ on it, one has an $m$-simplex in $K$, and hence the functor $\llbracket \lambda y.M  \rrbracket_{[-/(x)(i)]\rho}= G\circ \llbracket M \rrbracket_{[-/y][-/x(i)]\rho}$ is continuous.
		\end{enumerate}
\end{lem}

For the proof of Theorem \ref{theorem-lambda-model-syntatic}, we make the following remark. 

\begin{rem}
Just as the category $Set$ has enough points, the $\infty$-category $\mathscr{S}$ has enough points in the following sense: Let $f,g:X\rightarrow Y$ be functors between Kan complexes. If for each $x\in X$, $n\geq 0$ one has $f_{x}^n=g_{x}^n$, with $f_x^n:\pi_n(X,x)\rightarrow\pi_n(Y,f(x))$ and $g_x^n:\pi_n(X,x)\rightarrow\pi_n(Y,g(x))$ as maps induced by $f$ and $g$ respectively, then one has functorial equivalence $f\simeq g$. The property `$\,\mathscr{S}$ has enough points' can also be interpreted as: given a morphism $f:X\rightarrow Y$ in $\mathscr{S}$, if the induced map $f_x^n:\pi_n(X,x)\rightarrow\pi_n(Y,f(x))$ is an isomorphism of groups for each $n\geq 0$ and $x\in X$, then $f$ is a homotopy equivalence. 
\end{rem}

\begin{teor}\label{theorem-lambda-model-syntatic}
	Let $K$ be a reflexive Kan complex via the morphisms $F$, $G$, and  let $\mathfrak{M}=\langle K,\bullet,\llbracket \,\,\, \rrbracket\rangle$. Then
	\begin{enumerate}
		\item $\mathfrak{M}$ is a syntactic homotopic $\lambda$-model.
		\item $\mathfrak{M}$ is extensional iff there is a natural equivalence $\eta:id_K\rightarrow GF.$
	\end{enumerate}
\end{teor}
\begin{proof}
	1. The conditions in Definition \ref{homotopic-lambda-model-definition} (1), (2) are trivial. As to (3), given  $g\in K_n$,
	\begin{align*}
	\llbracket \lambda x.M \rrbracket_\rho\bullet g&=G(\boldsymbol{\lambda} f.\llbracket M \rrbracket_{[f/x]\rho})\bullet g \\
	&=F(G(\boldsymbol{\lambda} f.\llbracket M \rrbracket_{[f/x]\rho}))(g) \\
	&\xrightarrow{(\varepsilon_{\boldsymbol{\lambda} f.\llbracket M \rrbracket_{[f/x]\rho}})_g} (\boldsymbol{\lambda} f.\llbracket M \rrbracket_{[f/x]\rho}) (g) \\
	&=\llbracket M \rrbracket_{[g/x]\rho},
	\end{align*}
	where  $\varepsilon_{\boldsymbol{\lambda} f.\llbracket M \rrbracket_{[f/x]\rho}}$ is the natural equivalence, induced by $\varepsilon$, between the functors $F(G(\boldsymbol{\lambda} f.\llbracket M \rrbracket_{[f/x]\rho}),\boldsymbol{\lambda} f.\llbracket M \rrbracket_{[f/x]\rho}:K\rightarrow K_n$. Hence $(\varepsilon_{\boldsymbol{\lambda} f.\llbracket M \rrbracket_{[f/x]\rho}})_g$ is the equivalence induced by the $n$-simplex $g$ in $K$.
	
	The condition (4) is trivial, since if $\llbracket M \rrbracket_\rho\neq\llbracket M \rrbracket_\sigma
	 $ so there is $x\in FV(M)$ such that $\rho(x)\neq\sigma(x)$. The condition (5), given any vertex $a\in K$ and $y\notin FV(M)$
	\begin{align*}
	\boldsymbol{\lambda}f.\llbracket M(y) \rrbracket_{[f/y]\rho}&=\boldsymbol{\lambda}f.\llbracket M(x) \rrbracket_{[f/x]\rho}. 
	\end{align*}
	Applying $G$ and by Definition \ref{interpretation-definition} (c), it follows that
	\begin{align*}
	\llbracket \lambda y.M(y) \rrbracket_{\rho}&=G(\boldsymbol{\lambda}f.\llbracket M(y) \rrbracket_{[f/y]\rho}) \\
	&= G(\boldsymbol{\lambda}f.\llbracket M(x) \rrbracket_{[f/x]\rho}) \\
	&=\llbracket \lambda x.M(x) \rrbracket_{\rho} 
	\end{align*}
	Condition (6). By hypothesis, for every vertex $a\in K$, $n\geq 1$ and  $\omega\in \Omega_n(K,a)$
	\begin{align*}
	(\boldsymbol{\lambda}f.\llbracket M \rrbracket_{[f/x]\rho})(\omega)&=\llbracket M \rrbracket_{[\omega/x]\rho} \\
	&=\llbracket N \rrbracket_{[\omega/x]\rho} \\
	&=(\boldsymbol{\lambda}f.\llbracket N \rrbracket_{[f/x]\rho})(\omega),
	\end{align*}
	since $K$ does have enough points, then 
	$$\boldsymbol{\lambda}f.\llbracket M \rrbracket_{[f/x]\rho}=\boldsymbol{\lambda}f.\llbracket N \rrbracket_{[f/x]\rho}$$
	applying $G$ and by Definition \ref{interpretation-definition} (c),
	\begin{align*}
	\llbracket \lambda x.M \rrbracket_{\rho}&=G(\boldsymbol{\lambda}f.\llbracket M \rrbracket_{[f/x]\rho}) \\
	&= G(\boldsymbol{\lambda}f.\llbracket N \rrbracket_{[f/x]\rho}) \\
	&=\llbracket \lambda x.N \rrbracket_{\rho}
	\end{align*}
	2. Suppose that $\mathfrak{M}$ is extensional. Let $\omega\in\Omega_n(K,a)$. Then for all $a\in K$
	$$(GF(\omega))\bullet a=F(GF(\omega))(a)=((FG)F(\omega))(a)
	\xrightarrow{(\varepsilon_\omega F)_a} F(\omega)(a) 
	= \omega\bullet a,$$	
	by extensionality 
	$$GF(\omega)\xrightarrow{(\varepsilon_\omega F)} \omega= id_K(\omega),$$ 
	since $K$ does have enough points, hence $$GF\xrightarrow{\varepsilon F} Id_K.$$
	If $GF\xrightarrow{\eta} Id_K$. For all $\omega\in \Omega_n (K,a)$ by hypothesis and Definition \ref{interpretation-definition}
	$$F(a)(\omega)=a\bullet \omega\xrightarrow{\sigma_\omega} a'\bullet \omega= F(a')(\omega),$$
	since $K$ does have enough points, $Fa\xrightarrow{\sigma} Fa'$. Applying $G$, it follows that 
		$$a\xrightarrow{\eta_{a}}GF(a)\xrightarrow{G\sigma} GF(a')\xrightarrow{\tilde{\eta}_{a'}}a',$$
		where $\tilde{\eta}_{a'}$ is an inverse of $\eta_{a'}$. 
\end{proof}

\section{Homotopic $\lambda$-models}

Next, let us define cartesian closed $\infty$-category, points and paths of an $\infty$-category and  $\infty$-homotopic $\lambda$-model.    

\begin{defin}[Cartesian closed $\infty$-category]
	Let $\mathscr{C}$ be an $\infty$-category whose objects are small $\infty$-categories. We say that $\mathscr{C}$ is a cartesian closed $\infty$-category  (c.c.i.) if:
	\begin{enumerate}
		\item $\mathscr{C}$ has a terminal object $T$, i.e., $\mathscr{C}(X,T)$ is contractible for each $X\in\mathscr{C}$.
		\item For $X,Y\in\mathscr{C}$, there is the cartesian product $X\times Y$, and this belongs to $\mathscr{C}$,
		\item For $X,Y,Z\in\mathscr{C}$, there exists an internal morphism spaces $Y\Rightarrow Z$ in $\mathscr{C}$ such that sets the natural equivalence
		$$\mathscr{C}(X\times Y,Z)\simeq \mathscr{C}(X,Y\Rightarrow Z).$$  
	\end{enumerate} 	
\end{defin}

\begin{defin}[Enough points and $n$-paths]\label{Enough points and n-paths}
	Take an $\infty$-category $\mathscr{C}$ with a terminal object $T$. A point of an object $X$ is a morphism $x:T\rightarrow X$. The class of points of $X$ is denoted by $|X|_0$. 
	\begin{enumerate}
		\item We say that $\mathscr{C}$ does have enough points if for each pair of morphisms $f,g:X\rightarrow Y$ of $\mathscr{C}$ such that for each point $x:T\rightarrow X$ there is an equivalence $\sigma_x:f\circ x\simeq g\circ x$ in $\mathscr{C}(T,Y)$, then there is an equivalence $\sigma:f\simeq g$ in $\mathscr{C}(X,Y)$.
		\item An object $X\in\mathscr{C}$ does have enough points if one has (1) in the case that $Y=X$.
		\item Take the points $x, y\in |X|_0=\mathscr{C}(T,X)$. A 1-path $p:x\rightarrow y$ in $X$ is a 1-simplex at $\mathscr{C}(T,X)$. The class of the 1-paths of $X$ is denoted by $|X|_1=(\mathscr {C}(T,X))_1$. In general, for each $n\geq 1,$ the class of $n$-paths of $X$ corresponds to $|X|_n=(\mathscr{C}(T,X))_n$. 
	\end{enumerate}
\end{defin}

\begin{rem}\label{remark-infinity-categorical-strucutre}
Note that in Definition \ref{Enough points and n-paths} (3), since $\mathscr{C}$ is an $\infty$-category, all the 1-paths in $X$ (2-simplexes in $X$) are invertible. Then we say that $X\in\mathscr{C}$ has a `homotopic structure'. If $\mathscr{C}$ is an $\infty$-bicategory with a terminal object, we say that an object $X\in\mathscr{C}$ has `$\infty$-categorical structure'.   	
\end{rem}

\begin{defin}[Reflexive object]
	Let $\mathscr{C}$ be a c.c.i. An object $K\in\mathscr{C}$ is called reflexive if  $(K\Rightarrow K)$ is a weak retract of $K$ i.e., there are morphisms
	$$F:K\rightarrow(K\Rightarrow K),\hspace{1cm} G:(K\Rightarrow K)\rightarrow K$$  
	such that there is a natural equivalence $\varepsilon:FG\rightarrow id_{(K\Rightarrow K)}$. 
	
	\medskip If there is a natural equivalence $\eta:id_K\rightarrow GF$, then $K$ is an extensional object.	
\end{defin}

\begin{defin}[Homotopic $\lambda$-model]
	A  homotopic $\lambda$-model of a c.c.i. $\mathscr{C}$ is a quadruple $ \mathscr{K}=\langle K,F,G,\varepsilon\rangle$  where $K\in\mathscr{C}$ is a reflexive object via $F$, $G$ and $\varepsilon$ of the definition above. The quintuple $\mathscr{K}=\langle K,F,G,\varepsilon,\eta\rangle$ is an extensional homotopic $\lambda$-model, with $\eta$ being the natural equivalence of the same definition. 
\end{defin}

\begin{rem}
	By virtue of the Remark \ref{remark-infinity-categorical-strucutre}, if  $U$ is a reflexive object in a cartesian closed $\infty$-bicategory $\mathscr{C}$, we say that $K$ is an `$\infty$-categorical $\lambda$-model'. 
\end{rem}

\section{Kleisli $\infty$-categories} 

Next, we define the Kleisli structures on the $\infty$-categories, a general and direct version of those initially introduced by \cite{Hyland14} for the case of bicategories.

\begin{defin}[Kleisli structure]\label{Kleisli-strucutre}
	Let $\mathcal{K}$ be an $\infty$-category and $\mathcal{A}$ be an $\infty$-category contained in $\mathcal{K}$. A Kleisli structure $P$ on $\mathcal{A}\subseteq \mathcal{K}$ is the following.
	\begin{itemize}
		\item For each vertex $a\in\mathcal{A}$ an arrow $y_a:a\rightarrow Pa$ in $\mathcal{K}$.
		\item For each $a,b \in\mathcal{A}$ a functor
		$$\mathcal{K}(a,Pb)\rightarrow\mathcal{K}(Pa,Pb),\hspace{0.5cm}f\mapsto f^{\#}.$$
		\item A subcategory $\mathcal{K}^{L}\subseteq\mathcal{K}$ such that for all the vertices $a,b\in\mathcal{A}$, the homotopy equivalence
		$$\xymatrix{
			\mathcal{K}(a,Pb)\ar@<1ex>[r]^{(-)^\#} & \mathcal{K}^L(Pa,Pb).\ar@<1ex>[l]^{(-)y_a}
		}$$
	
		such for each horn $(g^\#,-,f^\#):\Lambda_1^2\rightarrow\mathcal{K}^L$  one has the equality of fibres 
		$$F_{g^\#,f^\#}^L=F_{g^\#,f^\#},$$
		where $F_{g^\#,f^\#}^L$ and $F_{g^\#,f^\#}$ are the fibres of the canonical maps $Fun(\Delta^2,\mathcal{K}^L)\rightarrow Fun(\Lambda^2_1,\mathcal{K}^L)$ and $Fun(\Delta^2,\mathcal{K})\rightarrow Fun(\Lambda^2_1,\mathcal{K})$ respectively, with $f:a\rightarrow Pb$ and $g:b\rightarrow Pc$ being edges in $\mathcal{K}$.

	\end{itemize}   
\end{defin}

It is clear that $P$ is a functor from  $\mathcal{A}$ to $\mathcal{K}$ such that for each 1-simplex $f:a\rightarrow b$ of $\mathcal{A}$,  sets $Pf=(y_bf)^{\#}:Pa\rightarrow Pb$.

\begin{ejem}
	The functor $P:Cat_\infty\rightarrow CAT_\infty$, given by $PA=[A^{op},\mathscr{S}]$, is a Kleisli structure on $Cat_\infty\subseteq CAT_\infty$.

	\medskip For each small $\infty$-category $A$, the arrow $y_A:A\rightarrow PA$ is the Yoneda embedding. For every functor $f:A\rightarrow PB$ there exists the functor $f^\#:PA\rightarrow PB$ which preserves small colimits. Besides, one has the categorical equivalence 
	$$\xymatrix{
		Fun(A,PB)\ar@<1ex>[r]^{(-)^\#} & Fun^L(PA,PB),\ar@<1ex>[l]^{(-)y_A}
	}$$
	where $Fun^L(PA,PB)$ is the $\infty$-category of functors which preserve small colimits \cite{DBLP:books/mk/Lurie}. If we take the subcategory $\mathscr{P}r^L\subseteq CAT_\infty$, whose objects are presentable $\infty$-categories and morphisms are functors which preserve small colimits, the categorical equivalence above is restricted to the homotopy equivalence  
	$$\xymatrix{
	CAT_\infty(A,PB)\ar@<1ex>[r]^{(-)^\#} & \mathscr{P}r^L(PA,PB).\ar@<1ex>[l]^{(-)y_A}
}$$
\end{ejem}

\begin{defin}[Kleisli $\infty$-category]
	Take a Kleisli structure $P$ on $\mathcal{A}\subseteq \mathcal{K}$. Define the Kleisli $\infty$-category $Kl(P)$ as follows. The objects  of $Kl(P)$ are the objects of $\mathcal{A}$ and the morphism spaces is defined by $$Kl(P)(a,b):=\mathcal{K}(a,Pb).$$
	for the all objects $a,b\in\mathcal{A}$.
	\end{defin}

\begin{rem}
	By Definition \ref{Kleisli-strucutre}, the homotopy equivalence $Kl(P)(a,b)\rightarrow\mathcal{K}^L(Pa,Pb),$ gets to establish that $KL(P)$ is an $\infty$-category embedded in $\mathcal{K}^L$.
	Another interesting way would be to define $Kl(P)$ as a weighted colimit or the pushout of diagram $$(Cat_\infty\subseteq CAT_\infty\xleftarrow{P}Cat_\infty)$$
	in the category of simplicial sets $Set_{\Delta}=Fun(\Delta^{op},Set)$ (complete and cocomplete) and so $Kl(P)$ would be an $\infty$-category.
\end{rem}

\begin{prop}
	$P(A\times B)\simeq PA\otimes PB$ in the $\infty$-category $\mathscr{P}r^L$.
\end{prop}
\begin{proof}
	\begin{align*}
Fun^L(P(A\times B),C)&\simeq Fun(A\times B,C) \\
	&\simeq Fun(A,C^B) \\
	&\simeq Fun^L(PA,Fun^L(PB,C)) 
    \end{align*}
	Thus, $\mathscr{P}r^L(P(A\times B),C)\simeq\mathscr{P}r^L(PA,PB\multimap C)\simeq\mathscr{P}r^L(PA\otimes PB,C).$
\end{proof}

\section{A cartesian closed $\infty$-category (c.c.i) with enough points}

In this section we prove that the Kleisli $\infty$-category $Kl(P)$ generated by the structure $P$ on $Cat_\infty\subseteq CAT_\infty$ is cartesian closed and has enough points. Thus $Kl(P)$ is a candidate for a higher $\lambda$-model.  

\begin{lem}
	The $\infty$-category $Kl(P)$ is cartesian closed.
\end{lem}
\begin{proof}
	\begin{align*}
	Fun(A\times B, PC)&\simeq Fun(A , [B,PC]) \\
	&= Fun(A , [B,{\mathscr{S}}^{C^{op}}]) \\
	&\simeq Fun(A , [B\times C^{op},\mathscr{S}]) \\
	&= Fun(A , P(B^{op}\times C)).
	\end{align*}
	Thus $Kl(P)[A\times B,C]\simeq Kl(P)(A,B^{op}\times C)=Kl(P)(A,B\Rightarrow C)$.
\end{proof}

\begin{teor}\label{points-enough-theorem}
	The $\infty$-category $Kl(P)$ does have enough points. 	
\end{teor}
\begin{proof}
	A morphism $A\rightarrow B$ in $Kl(P)$ corresponds to a functor $A\rightarrow PB$. Since $PA$ is a closure of $A$ under small colimits, such a functor corresponds to a small
 colimit
	preserving functor $PA\rightarrow PB$. Since $PA$ and $PB$ are weakly contractible, the functor $PA\rightarrow PB$ is sufficiently determined by all the vertices of $PA$ (points of $A$). 
\end{proof}

In \cite{Martinez21}, one can find some examples of reflexive objects in the category $KL(P)$, which are provided by methods of solving domain equations on arbitrary cartesian closed $\infty$-categories, where these types of equations are called  \textit{Homotopy Domain Equations}.

\section{Conclusion}

What we have done here is a beginning for the construction of a Homotopy Domain Theory (HoDT) which provides techniques to build homotopy $\lambda$-models that allow for the generalization of Church-like conversion relations (such as, e.g., $\beta$-equality, $\eta$-equality)  to higher term-contraction induced equivalences. 

\medskip Besides, we generalized the Kleisli bicategory to a Kleisli $\infty$-category, and we show that it is closed cartesian with enough points.  For future work, we could apply the techniques of HoDT to this Kleisli $\infty$-category and thus obtain a reflexive Kan complex (homotopic $\lambda$-model) with relevant information. 

\medskip On the other hand, we define the interpretation of the $\beta\eta$-contractions in a reflexive Kan complex, whose $\infty$-groupoid structure induces higher $\beta\eta$-contractions which  would inhabit a type of identity (based on computational paths).  This work could be seen as the beginning of the semantics of another version,  based on computational paths, of the Theory of Homotopy Types.

\bibliography{mybibfile}

\end{document}